\documentclass[letterpaper,10 pt, conference]{ieeeconf}  

\IEEEoverridecommandlockouts                              
                                                          
\usepackage[utf8]{inputenc}
\usepackage{include_packages}

\title{\LARGE \bf When showing your hand pays off: \\Announcing strategic intentions in Colonel Blotto games \\
\thanks{This work is supported by ONR grant \#N00014-17-1-2060, NSF grant \#ECCS-1638214, and UCOP Grant LFR-18-548175.}}
\author{Rahul Chandan \and Keith Paarporn \and Jason R. Marden
\thanks{R. Chandan, K. Paarporn and J. R. Marden are with the Department of Electrical and Computer Engineering and the Center of Control, Dynamical Systems and Computation, UC Santa Barbara, USA. Email: \href{mail_to:rchandan@ucsb.edu}{rchandan@ucsb.edu}, \href{mail_to:kpaarporn@ucsb.edu}{kpaarporn@ucsb.edu}, \href{mail_to:jrmarden@ece.ucsb.edu}{jrmarden@ece.ucsb.edu}.}}
\date{September 2019}

\begin{document}

\maketitle

\begin{abstract}
     In competitive adversarial environments, it is often advantageous to obfuscate one's strategies or capabilities. However, revealing one's strategic intentions may shift the dynamics of the competition in complex ways. Can it ever be advantageous to reveal strategic intentions to an opponent? In this paper, we consider three-stage Colonel Blotto games in which one  player can choose whether or not to pre-commit resources to a single battlefield before play begins. This pre-commitment is public knowledge. In response, the opponent can either secure the battlefield by matching the pre-commitment with its own forces, or withdraw. In a two-player setting, we show that a weaker player never has an incentive to pre-commit any amount of resources to a battlefield regardless of how valuable it is. We then consider a three-player setting in which two players fight against a common adversary on separate fronts. Only one of the two players facing the adversary has the option of pre-committing. We find there are instances where this player benefits from pre-committing. The analysis indicates that under non-cooperative team settings and no possibility of forming alliances, there can be incentives to publicly announce one's strategic intentions to an adversary. 
\end{abstract}

\section{Introduction}
The Colonel Blotto game models competitive resource allocation problems in which two opposing players strategically distribute their respective budgets across a set of battlefields. The aim of each player is to accrue as much value as possible by securing battlefields. The game has been studied in the context of political campaigns \cite{Behnezhad_2018,Thomas_2018}, security of cyber-physical systems \cite{Ferdowsi_2017,Gupta_2014a}, and competitive markets \cite{Golman_2009}.

First introduced by Borel in 1921 \cite{Borel}, mixed strategy Nash equilibria have been developed and characterized for Blotto games over the last one hundred years. These equilibria describe players' optimal security strategies as well as their performance guarantees. The work of Roberson \cite{Roberson_2006} established a general solution for an arbitrary number of homogeneous battlefields and asymmetric resource budgets. Solutions to the General Lotto game, a variant of the Colonel Blotto game, have been established in more general settings as it allows for a higher degree of analytical tractability \cite{Hart_2008,Kovenock_2015}. Irrespective of the setup, the non-existence of pure Nash equilibria in many cases of interest suggests competitors that employ predictable strategies are exploitable \cite{Golman_2009,Schwartz_2014}.

However, revealing one's intentions may shift the dynamics of the competition. For example, consider two  firms A and B competing for control over multiple markets. Firm B expresses interest in a new, niche market, and announces its investment in this market. This presents a choice to firm A; A can choose either to 1) give up on this market and focus its resources on others, possibly in competition with firms other than B, or 2) divert some of its resources to stay competitive against B in the new market. In effect, firm B sets the price to compete in the new market. 

In the context of Colonel Blotto games, the focus of this paper centers on whether a player can improve upon its standard equilibrium performance by announcing a pre-commitment of resources to a single battlefield. We first consider a three-stage, two-player scenario. In the first stage, one player decides whether or not to pre-commit resources to a battlefield of its choice. If it does not commit, the players play the classic one-shot game over the entire set of battlefields in stage two. If it decides to commit, the player decides how many resources to allocate to that battlefield. Then, in stage two, the adversary decides to either match the pre-commitment on that battlefield, or to send no resources. In stage three, they play the one-shot game over the remaining battlefields with the re-distribution of forces.

We then consider a three-stage, three-player game in which two players compete against a common adversary, each on separate fronts. Only one of the players on the team (say, player 1) has the option of pre-committing to a battlefield in the first stage. In this scenario, the adversary must now decide how to split its forces among two fronts, as well as whether or not to call the pre-commitment of player 1. In stage three, two one-shot games are played between the adversary and each of the two team members. 

Previous work has considered similar three-stage, three-player setups, in which the team players have the option of transferring resources among each other, or of adding battlefields to their respective one-shot games \cite{Kovenock_2010,Kovenock_2012,Gupta_2014a,Gupta_2014b}. The problem we pose here differs in regard to the non-possibility of forming alliances even when there is a common enemy. This reflects non-cooperative settings characterized by a lack of communication or consent to collaboration \cite{Radner_1962,Tsitsiklis_1985,Liu_2004,Grimsman_2018}. For example, consider two different companies that must protect their data servers against a hacker intent on accessing as much data as possible (from one or both of the companies' servers). The security of one of the company's servers is not the responsibility or even concern of the other. If one of the companies had an opportunity to announce a pre-commitment of security resources on a portion of its servers, how would this change the hacker's own allocation of resources to access data? We seek to address whether such a pre-commitment relieves attack pressure the company would have faced otherwise.

\subsection*{Our Contributions:}
Our main findings are as follows. In a one-vs-one setting, the weaker resource player will never have an incentive to pre-commit resources. That is, it will always prefer to play the standard one-shot game contested over the entire set of battlefields. In the three-player setup, we find there are instances in which the weaker player can benefit from pre-committing. 

These findings provide insight into the circumstances under which a single competitor can profit from a public announcement of strategies. A two-player setting is not complex enough for a pre-commitment to sufficiently draw an adversary's attention away from valuable prizes. However in a three-player setting, the presence of a team member can serve as a distraction to the adversary.

\section{Preliminaries: Static Colonel Blotto and General Lotto Games}

Consider a game wherein two players, $A$ (the adversary) and $B$, simultaneously allocate their respective budgets, $X_A$ and $X_B$, over a set of $n \geq 2$ battlefields, $\B = \{1, \dots, n\}$. 
Each battlefield $b \in \B$ has an associated value $v_b > 0$, which is won by the player assigning a higher level of force to $b$, at the cost of the opposing player, i.e., given budget allocations $\mathbf{x}_i = (x_{i,1}, \dots, x_{i,n}) \in \mathbb{R}^n$ for both $i \in \{A,B\}$, the payoff to player $i$ is
\[ \sum^n_{b=1} v_b \, \mathrm{sign}(x_{i,b} - x_{-i,b}), \]
where $-i$ is the player opposing $i$. 
For ease of notation, we will assume that player $A$ wins if there is a tie.

The \textit{Colonel Blotto game} is a game as defined above where, for each player $i \in \{A,B\}$, a valid budget allocation $\mathbf{x}_i$ satisfies $x_{i,b} \geq 0$ for all $b \in \B$, and $\sum^n_{b=1} x_{i,b} \leq X_i$. A valid mixed strategy in a Colonel Blotto game is an $n$-variate probability distribution function, $S_i: \mathbb{R}^n \to [0,1]$, with support only in the set of valid budget allocations.

The \textit{General Lotto game} represents a variation of the Colonel Blotto game where a valid mixed strategy $S_i$ need only satisfy the budget constraint in expectation, i.e., for both $i \in \{A,B\}$,
\[  \sum^n_{b=1} \mathbb{E}_{S_{i,b}} [ x_{i,b} ] \leq X_i, \]
where $S_{i,b}:\mathbb{R} \to [0,1]$ is the univariate marginal distribution of $S_i$ corresponding to battlefield $b \in \B$.
Note that the budget constraint is relaxed in the General Lotto game. Indeed, a budget allocation $\mathbf{x}_i$ such that $\sum^n_{b=1} x_b > X_i$ can still be attributed positive probability in the mixed strategy $S_i$ of a player $i \in \{A,B\}$.
Given player budgets $X_A, X_B > 0$, and a set of $n$ battlefields $\B$ with $\mathbf{v} = (v_1, \dots, v_n)$, we denote the corresponding General Lotto game as $\gl(X_A,X_B,\B)$. 

As both Colonel Blotto and General Lotto are zero-sum games, the players' equilibrium payoffs are unique. We reproduce the equilibrium payoffs of the General Lotto game here.
\begin{prop}[General Lotto payoffs \cite{Kovenock_2012}]
    For player budgets $X_A, X_B > 0$, and $\phi = \sum_{b=1}^n v_b$, the equilibrium payoff to a player $i \in \{A,B\}$ in the game $\gl(X_A,X_B,\B)$ is,
    \begin{equation} \label{eq:twoplayer_payoff}
        \begin{cases} 
            \phi \left( \frac{X_i}{X_{-i}} - 1 \right)& \quad \text{if } 0 < X_i \leq X_{-i} \\
            \phi \left( 1 - \frac{X_{-i}}{X_i} \right)& \quad \text{if } X_i > X_{-i}.
        \end{cases} 
    \end{equation}
\end{prop}
The players' equilibrium payoffs in the General Lotto game  arbitrarily approximates those of the Colonel Blotto game when the number of battlefields becomes  large (see \cite{Kovenock_2012} for a discussion), and are relatively simple to compute. For this reason, we consider the static General Lotto game as the primitive model of competitive resource allocation in the forthcoming multi-stage models.

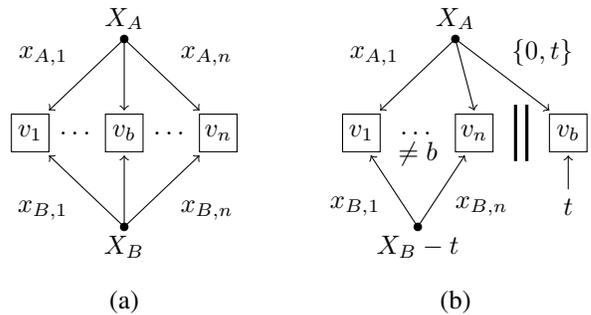
\begin{figure}
    \centering
    \begin{minipage}[b]{0.2\textwidth}
    \begin{tikzpicture}
        \draw (0.5,0) rectangle (1,0.5);
        \draw (1.375,0.25) node {$\dots$};
        \draw (1.75,0) rectangle (2.25,0.5);
        \draw (2.625,0.25) node {$\dots$};
        \draw (3,0) rectangle (3.5,0.5);
        \node at (0.75,0.25) {$v_1$};
        \node at (2,0.25) {$v_b$};
        \node at (3.25,0.25) {$v_n$};
        \draw[fill] (2,-1) circle [radius=0.05];
        \draw[fill] (2,1.5) circle [radius=0.05];
        \node[above] at (2,1.5) {$X_A$};
        \node[below] at (2,-1) {$X_B$};
        \draw[->] (2,-1) -- (1,-0.05);
        \node[below left] at (1.375,-0.525) {$x_{B,1}$};
        \draw[->] (2,-1) -- (2,-0.05);
        \draw[->] (2,-1) -- (3,-0.05);
        \node[below right] at (2.625,-0.525) {$x_{B,n}$};
        \draw[->] (2,1.5) --  (1,0.55);
        \node[above left] at (1.375,1.025) {$x_{A,1}$};
        \draw[->] (2,1.5) --  (2,0.55);
        \draw[->] (2,1.5) --  (3,0.55);
        \node[above right] at (2.625,1.025) {$x_{A,n}$};
        \node at (2,-2) {(a)};
    \end{tikzpicture}
    \end{minipage}
    \hspace{10pt}
    \begin{minipage}[b]{0.2\textwidth}
    \begin{tikzpicture}
        \draw (0.5,0) rectangle (1,0.5);
        \draw (1.5,0.25) node {$\dots$};
        \node[below] at (1.5,0.25) {$\neq b$};
        \draw (2,0) rectangle (2.5,0.5);
        \draw[very thick] (2.8,-0.125) -- (2.8,0.625);
        \draw[very thick] (2.95,-0.125) -- (2.95,0.625);
        \draw (3.25,0) rectangle (3.75,0.5);
        \node at (0.75,0.25) {$v_1$};
        \node at (2.25,0.25) {$v_n$};
        \node at (3.5,0.25) {$v_b$};
        \draw[fill] (1.5,-1) circle [radius=0.05];
        \draw[fill] (2,1.5) circle [radius=0.05];
        \node[above] at (2,1.5) {$X_A$};
        \node[below] at (1.5,-1) {$X_B-t$};
        \node[below] at (3.5,-0.5) {$t$};
        \draw[->] (1.5,-1) -- (0.875,-0.05);
        \node[below left] at (1.125,-0.475) {$x_{B,1}$};
        \draw[->] (1.5,-1) -- (2.125,-0.05);
        \node[below right] at (1.875,-0.475) {$x_{B,n}$};
        \draw[->] (3.5,-0.5) -- (3.5,-0.05);
        \draw[->] (2,1.5) --  (1,0.55);
        \node[above left] at (1.375,1.025) {$x_{A,1}$};
        \draw[->] (2,1.5) --  (2.25,0.55);
        \draw[->] (2,1.5) --  (3.25,0.55);
        \node[above right] at (2.625,1.025) {$\{0,t\}$};
        \node at (2,-2) {(b)};
    \end{tikzpicture}
    \end{minipage}
    \caption{The two possible outcomes of the General Lotto game with pre-commitments. 
    In a), player $B$ does not pre-commit to a transfer of its forces to a battlefield $b \in \B$. The two players play a one-shot General Lotto game over all the battlefields in $\B$. 
    In b), player $B$ pre-commits a force $t \in [0,X_B]$ to a battlefield $b \in \B$. Player $A$ subsequently has the choice of sending an equivalent force to battlefield $b$ (calling), or abstaining from allocating any of its budget to $b$ (folding). The two players play a one-shot General Lotto game over the rest of the battlefields, $\B \setminus \{b\}$.}
    \label{fig:two-player-games}
\end{figure}

\section{Pre-commitments in a Two-Player Setting}

In this section, we formulate a three-stage, two-player General Lotto game.  
In the first stage, player $B$ has the option to pre-commit resources to a single battlefield of its choice. 
Here, we are assuming player $B$ has fewer resources than player $A$, i.e., $X_B < X_A$. 
If $B$ decides not to commit any resources, the players play their equilibrium strategies in the static General Lotto game $\gl(X_A,X_B,\B)$ in the second stage. 
In the final stage, players are assigned the final payoffs according to \eqref{eq:twoplayer_payoff}. 
This outcome is illustrated in Figure \ref{fig:two-player-games}(a). 

However, if $B$ does decide to pre-commit, it chooses one of the battlefields $b \in \B$ as well as an amount $t\in[0,X_B]$ of resources to commit to $b$ in the first stage. 
In the second stage, player $A$  observes player $B$'s choices $b$ and $t$, and decides whether to dedicate $t$ of its own budget $X_A$ in order to win the battlefield $b$ ($A$ \textit{calls}), or to save its budget for the remaining battlefields and lose the battlefield $b$ ($A$ \textit{folds}).
If $A$ calls, it wins battlefield $b$, and the static Lotto game $\gl(X_A-t,X_B-t,\B\setminus b)$ is played among the $n-1$ remaining battlefields in stage three. Player $A$'s final payoff is then
\begin{equation}\label{eq:call_payoff}
    u_A^{\text{call}}(X_A,X_B,t) = (\phi-v_b)\left(1-\frac{X_B-t}{X_A-t} \right) + v_b.
\end{equation}
If $A$ folds, battlefield $b$ is won by player $B$, and the static Lotto game $\gl(X_A-t,X_B,\B\setminus b)$ is played among the other $n-1$ battlefields. Player $A$'s final payoff in this case is
\begin{equation}\label{eq:fold_payoff}
    u_A^{\text{fold}}(X_A,X_B,t) = (\phi-v_b)\left(1-\frac{X_B-t}{X_A} \right) - v_b.
\end{equation}
In effect, the competition for $b$ becomes a separate, single battlefield Colonel Blotto game on which $B$ has committed to sending exactly $t$ resources. The setup here is illustrated in Figure \ref{fig:two-player-games}(b). We assume player $A$ is rational and chooses the action in stage 2 that maximizes its final payoff. Formally, we may then write the final (optimal) payoff to player $A$ as 
\begin{equation}
    \pi_A^t(X_A,X_B) := \max\left\{ u_A^{\text{call}}(X_A,X_B,t),u_A^{\text{fold}}(X_A,X_B,t) \right\}.
\end{equation}
As the game remains zero-sum, the player $B$ receives a final payoff of $\pi_B^t(X_A,X_B) := -\pi_A^t(X_A,X_B)$. 
We say that there exists a profitable pre-commitment for $B$ if there exist $b\in\B$ and $t\in[0,X_B]$ such that
\begin{equation}\label{eq:improvement}
    \Delta\pi_B^t := \pi_B^t(X_A,X_B) - \pi_B(X_A,X_B) > 0,
\end{equation}
where $\pi_B(X_A,X_B)$ is given by the first entry of \eqref{eq:twoplayer_payoff}. 
Intuitively, if \eqref{eq:improvement} holds for some $b\in\B$ and $t\in[0,X_B]$, player $B$ can improve upon its payoff from the standard General Lotto game setup by making an appropriate pre-commitment. 
The following theorem states there is \emph{never} an incentive for $B$ to make a pre-commitment. In other words, condition \eqref{eq:improvement} never holds for any $X_B < X_A$, $t \in [0,X_B]$, and $b \in \B$.

\begin{thm}\label{theorem:twoplayer}
    Consider the two-player sequential General Lotto game where $X_B < X_A$. Then there is no pre-commitment $t \in [0,X_B]$ for player $B$ to any battlefield $b \in \B$ such that $\Delta\pi_B^t \geq 0$.
\end{thm}
The proof is provided in Appendix A.

\begin{rem}
    It is worth noting here that the above result cannot be shown by applying a min-max argument.
    By making a pre-commitment of $t$ forces to battlefield $b$, player $B$ does not only impose a static allocation in its own feasible set of actions, it also reduces the set of admissible actions available to player $A$.
    Observe that $A$ continues to play the standard General Lotto game against $B$ on battlefields $\B \setminus b$, but that its actions on $b$ are restricted to either call or fold.
    If $A$ were still capable of playing its original mixed strategy across all battlefields in $\B$ after $B$'s pre-commitment, \textit{then} a min-max argument would be applicable.
\end{rem}

\section{The Three-Player Setting}
In the previous section, we showed that no pre-commitment can increase the payoff of a weaker player in the two-player setting.
Since the players are solely focused on competing with each other, the weaker player has no recourse that will increase its payoff, as it is already at a disadvantage.
In this section, we consider whether a weaker player can increase its payoff via pre-commitment in a setting where the adversary is competing with an additional player.

\subsection{Coalitional General Lotto games with pre-commitment}
We consider a modified version of the three-player coalitional General Lotto game studied in, e.g., \cite{Kovenock_2010,Kovenock_2012,Gupta_2014a,Gupta_2014b}.
In the game, a player $A$ competes against two players, $1$ and $2$, on two disjoint sets of battlefields, $\B_1$ and $\B_2$, with battlefield values $\mathbf{v}_j = (v_{j,1}, \dots, v_{j,|\B_j|})$ such that $\phi_j = \sum_{b \in \B_j} v_{j,b}$, for $j = 1, 2$.
Players 1, 2 and $A$ have positive budgets $X_1$, $X_2$, and $X_A$, respectively. 
Without loss of generality, we normalize the players' budgets such that $X_A = 1$.
The game has the following sequential structure:
\begin{enumerate}[leftmargin=*]
    \item Player 1 decides whether to pre-commit a force $t \in [0,X_1]$ that it will not use, to pre-commit a force $t \in [0,X_1]$ to a battlefield $b \in \B_1$, or to not make any pre-commitment. We denote by $X^t_1$, player 1's budget after making its decision, i.e., $X^t_1 = X_1 - t$ if player 1 makes a pre-commitment $t \in [0,X_1]$, and $X^t_1 = X_1$ if it does not.
    \item If player 1 chose not to pre-commit a force $t$ to any battlefield $b \in \B_1$, player $A$ splits its budget between the two sets of battlefields; a force $X_{A1} \geq 0$ is assigned to $\B_1$, and a force $X_{A2} \geq 0$ is assigned to $\B_2$ such that $X_{A1} + X_{A2} = 1$, as in \cite{Kovenock_2012}. 
    
    Else, if player 1 pre-committed a force $t$ to a battlefield $b \in \B_1$, then player $A$ splits its budget such that $X^t_{A1} + X^t_{A2} = 1$.
    If $X^t_{A1} > t$ and $A$ calls, the forces $X^t_{A1}-t, X_{A2}$ are assigned to sets $\B_1 \setminus \{b\}$ and $\B_2$, respectively.
    Otherwise, $A$ folds, and $X^t_{A1}, X_{A2} \geq 0$ are assigned to sets $\B_1 \setminus \{b\}$ and $\B_2$, respectively.
    \item If player 1 pre-committed a force $t$ to a battlefield $b \in \B_1$, the players play two one-shot General Lotto games; the game $\gl(X^t_1,X^t_{A1},\B_1 \setminus \{b\})$, and $\gl(X_2,X^t_{A2},\B_2)$. Simultaneously, player $A$ wins battlefield $b$ if it called in the previous step, or loses $b$ if it folded. 
    
    Else, if player 1 did not pre-commit a force to a battlefield in $\B_1$, the players play two one-shot General Lotto games; $\gl(X_1,X_{A1},\B_1)$, and $\gl(X_2,X_{A2},\B_2)$.
\end{enumerate}

The decision made in each step becomes common knowledge among the three players in the subsequent steps. 
The payoff to player $A$ is the sum of the payoffs it receives in the one-shot General Lotto games it plays with each of the players $i \in \{1,2\}$, plus its losses or gains on the battlefield $b$ if player $1$ pre-commits.
Player $A$'s decision to call or fold on a pre-commitment always maximizes its payoff, as in the previous section.
Observe that if player 1 decides to use less than its full budget of $X_1$, or to pre-commit a force $t \in [0, X_1]$ to a given battlefield $b \in \B_1$, it does not require any dialogue with player 2.
The game structures ensuing from player 1's choices of not pre-committing any force, and of pre-committing a force $t$ to the battlefield $b \in \B_1$ are illustrated in \cref{fig:coalitional_blotto} (a) and (b), respectively.

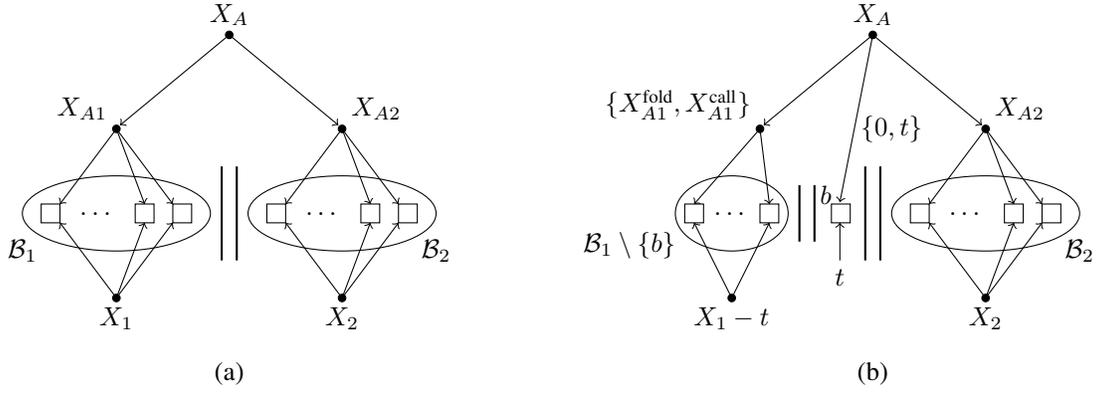
\begin{figure*}[t]
    \centering
    \begin{tikzpicture}
        \filldraw[color=black, fill=white] (1,0.125) ellipse (1.25 and 0.5);
        \node at (-0.25,-0.375) {$\B_1$};
        \filldraw[color=black, fill=white] (4,0.125) ellipse (1.25 and 0.5);
        \node at (5.25,-0.375) {$\B_2$};
        \draw (0,0) rectangle (0.25,0.25);
        \draw (0.75,0.125) node {$\dots$};
        \draw (1.25,0) rectangle (1.5,0.25);
        \draw (1.75,0) rectangle (2,0.25);
        \draw[thick] (2.4,-0.5) -- (2.4,0.75);
        \draw[thick] (2.6,-0.5) -- (2.6,0.75);
        \draw (3,0) rectangle (3.25,0.25);
        \draw (3.75,0.125) node {$\dots$};
        \draw (4.25,0) rectangle (4.5,0.25);
        \draw (4.75,0) rectangle (5,0.25);
        \draw[fill] (1,-1) circle [radius=0.05];
        \draw[fill] (4,-1) circle [radius=0.05];
        \draw[fill] (2.5,2.5) circle [radius=0.05];
        \draw[fill] (1,1.25) circle [radius=0.05];
        \draw[fill] (4,1.25) circle [radius=0.05];
        \draw[->] (2.5,2.5) -- (1.05,1.3);
        \draw[->] (2.5,2.5) -- (3.95,1.3);
        \draw[->] (1,-1) -- (0.25,0);
        \draw[->] (1,-1) -- (1.375,0);
        \draw[->] (1,-1) -- (1.75,0);
        \draw[->] (1,1.25) -- (0.25,0.25);
        \draw[->] (1,1.25) -- (1.375,0.25);
        \draw[->] (1,1.25) -- (1.75,0.25);
        \draw[->] (4,-1) -- (3.25,0);
        \draw[->] (4,-1) -- (4.375,0);
        \draw[->] (4,-1) -- (4.75,0);
        \draw[->] (4,1.25) -- (3.25,0.25);
        \draw[->] (4,1.25) -- (4.375,0.25);
        \draw[->] (4,1.25) -- (4.75,0.25);
        \node[above left] at (1,1.25) {$X_{A1}$};
        \node[above right] at (4,1.25) {$X_{A2}$};
        \node[below] at (1,-1) {$X_1$};
        \node[below] at (4,-1) {$X_2$};
        \node[above] at (2.5,2.5) {$X_A$};
        \node at (2.5,-2) {(a)};
    \end{tikzpicture}
    \qquad \qquad
    \begin{tikzpicture}
        \filldraw[color=black, fill=white] (0.625,0.125) ellipse (0.75 and 0.5);
        \filldraw[color=black, fill=white] (4,0.125) ellipse (1.25 and 0.5);
        \node at (5.25,-0.375) {$\B_2$};
        \draw (0,0) rectangle (0.25,0.25);
        \draw (0.625,0.125) node {$\dots$};
        \node[below left] at (0,0) {$\B_1 \setminus \{b\}$};
        \draw (1,0) rectangle (1.25,0.25);
        \draw[thick] (1.525,-0.25) -- (1.525,0.5);
        \draw[thick] (1.725,-0.25) -- (1.725,0.5);
        \draw (1.95,0) rectangle (2.2,0.25);
        \node[above left] at (2.075,0.125) {$b$};
        \draw[thick] (2.4,-0.5) -- (2.4,0.75);
        \draw[thick] (2.6,-0.5) -- (2.6,0.75);
        \draw (3,0) rectangle (3.25,0.25);
        \draw (3.75,0.125) node {$\dots$};
        \draw (4.25,0) rectangle (4.5,0.25);
        \draw (4.75,0) rectangle (5,0.25);
        \draw[fill] (0.625,-1) circle [radius=0.05];
        \draw[fill] (4,-1) circle [radius=0.05];
        \draw[fill] (2.5,2.5) circle [radius=0.05];
        \draw[fill] (1,1.25) circle [radius=0.05];
        \draw[fill] (4,1.25) circle [radius=0.05];
        \draw[->] (2.5,2.5) -- (1.05,1.3);
        \draw[->] (2.5,2.5) -- (3.95,1.3);
        \draw[->] (0.625,-1) -- (0.125,0);
        \draw[->] (0.625,-1) -- (1.125,0);
        \draw[->] (2.0625,-0.5) -- (2.0625,0);
        \node[below] at (2.0625,-0.5) {$t$};
        \draw[->] (1,1.25) -- (0.125,0.25);
        \draw[->] (1,1.25) -- (1.125,0.25);
        \draw[->] (2.5,2.5) -- (2.0625,0.275);
        \node at (2.75,1.25) {$\{0,t\}$};
        \draw[->] (4,-1) -- (3.25,0);
        \draw[->] (4,-1) -- (4.375,0);
        \draw[->] (4,-1) -- (4.75,0);
        \draw[->] (4,1.25) -- (3.25,0.25);
        \draw[->] (4,1.25) -- (4.375,0.25);
        \draw[->] (4,1.25) -- (4.75,0.25);
        \node[above left] at (1,1.25) {$\{X_{A1}^\text{fold},X_{A1}^\text{call}\}$};
        \node[above right] at (4,1.25) {$X_{A2}$};
        \node[below] at (0.625,-1) {$X_1-t$};
        \node[below] at (4,-1) {$X_2$};
        \node[above] at (2.5,2.5) {$X_A$};
        \node at (2.5,-2) {(b)};
    \end{tikzpicture}
    \caption{
    (a) The game structure resulting from player 1's choice to not pre-commit a force to any battlefield in $\B_1$ (Stage 1). In Stage 2, player $A$ splits its budget $X_A$ into the two forces $X_{A1}, X_{A2} > 0$ such that $X_{A1} + X_{A2} = X_A$. In the final stage, two General Lotto games are played, $\gl(X_1,X_{A1},\B_1)$, and $\gl(X_2,X_{A2},\B_2)$, between $A$ and 1, and $A$ and 2, respectively.
    (b) The game structure resulting from player 1's choice to pre-commit a force $t \in [0,X_1]$ to battlefield $b \in \B_1$ (Stage 1). 
    In Stage 2, player $A$ decides whether to match player 1's pre-commitment or fold (i.e. $x_{A,b} \in \{0,t\}$, and splits the remainder of its budget into two forces $X_{A1}, X_{A2} > 0$. If $A$ calls, then $X_{A1} + X_{A2} = 1-t$, otherwise, $X_{A1} + X_{A2} = 1$. 
    In the final stage, two General Lotto games are played, $\gl(X^t_1,X_{A1},\B_1 \setminus \{b\})$, and $\gl(X_2,X_{A2},\B_2)$, between $A$ and 1, and $A$ and 2, respectively. Player $A$ wins the battlefield $b \in \B_1$ if it calls, and loses $b$ if it folds.}
    \label{fig:coalitional_blotto}
\end{figure*}

From \cite{Kovenock_2012}, the optimal split of player $A$'s budget into $X_{A1}, X_{A2}$ in stage 2 is determined by distinct regions of the parameter space with $X_i, X_{-i} > 0$ for $i \in \{1,2\}$, $\phi_1$ and $\phi_2$, where $-i$ is the other player in $\{1,2\}$:
\begin{itemize}[leftmargin=*]
    \item {\bf Case 1:} $\frac{\phi_i}{\phi_{-i}} > \frac{ \max\{ (X_i)^2, 1 \} }{X_i X_{-i}}$. This corresponds to the region where player $A$ maximizes its payoff by committing all of its forces to the General Lotto game on the set of battlefields $\B_i$, i.e., $X_{Ai} = X_A$ and $X_{A(-i)} = 0$.
    \item {\bf Case 2:} $\frac{\phi_i}{\phi_{-i}} > \frac{X_i}{X_{-i}}$ and $0 < 1 - \sqrt{\frac{\phi_i X_i X_{-i}}{\phi_{-i}}} \leq X_{-i}$. This is the range of values for which the following split maximizes $A$'s payoff: 
    \begin{align*}
        X_{Ai} &= \sqrt{\frac{\phi_i X_i X_{-i}}{\phi_{-i}}} > X_i \\
        X_{A(-i)} &= 1-X_{Ai} < X_{-i}.
    \end{align*}
    \item {\bf Case 3:} $\frac{\phi_i}{\phi_{-i}} \geq \frac{X_i}{X_{-i}}$ and $1 - \sqrt{\frac{\phi_i X_i X_{-i}}{\phi_{-i}}} > X_{-i}$. In order to maximize its payoff, player $A$ splits its budget as 
    \[ X_{Ai} = \frac{\sqrt{X_i \phi_i}}{\sqrt{X_i \phi_i} + \sqrt{X_{-i} \phi_{(-i)}}}, \quad \forall i \in \{1,2\}. \]
    \item {\bf Case 4:} $\frac{\phi_i}{\phi_{-i}} = \frac{X_i}{X_{-i}}$ and $X_i + X_{-i} \geq 1$. Any split of player $A$'s budget results in the same payoff, thus $X_{Ai} \in [0,1]$, $X_{A(-i)} = 1 - X_{Ai}$.
\end{itemize}
Note that player $A$ splits its budget differently in the cases 1 and 2 depending on whether $i = 1$ or $i = 2$. Meanwhile, the split is independent of $i$ in cases 3 and 4. The regions are illustrated in \cref{fig:cases} for $\phi_1 = \phi_2$.
If player 1 pre-commits in stage 1 of the game, and player $A$ calls, we will be re-normalizing the players' budgets such that $X_A - t = 1$ to retain the cases for player $A$'s optimal budget divisions, i.e., for a pre-commitment $t \in [0,X_1]$, if player $A$ calls, $X_1 - t \to (X_1 - t)/(1-t)$, $X_2 \to X_2/(1-t)$, and $X_A - t \to 1$.

\begin{figure}[t]
    \centering
    \begin{tikzpicture}[scale=2.5]
        \draw[very thin, gray!30, step=0.25 cm](0,0) grid (2,2);
        \draw[<->] (0,2) node[above, left] {$X_2$} -- (0,0) -- (2,0) node[below, right] {$X_1$};
        \foreach \x in {0,0.5,...,1.5}
            \draw[xshift=\x cm, thick] (0pt,-0.5pt)--(0pt,0.5pt) node[below] {$\x$};
        \foreach \x in {0,0.5,...,1.5}
            \draw[yshift=\x cm, thick] (-0.5pt,0pt)--(0.5pt,0pt) node[left] {$\x$};
        \draw[dashed] (1,0) -- (1,1) -- (0,1);
        \draw[very thick,domain=0.5:2] plot (\x,\x);
        \draw[domain=0.5:1] plot (\x, {(1-\x)^2/\x});
        \draw[domain=1:2] plot (\x, {1/\x});
        \draw (1.625,1) node {Case 1 (i=2)};
        \draw (1.5,0.25) node {Case 2 (i=2)};
        \draw[domain=0.5:1] plot ({(1-\x)^2/\x}, \x);
        \draw[domain=1:2] plot ({1/\x}, \x);
        \draw (1.125,1.75) node {Case 1 (i=1)};
        \draw (0.375,1.25) node {Case 2 (i=1)};
        \draw (0.25,0.25) node {Case 3};
        \draw (1.75,1.5) node {\bf Case 4};
    \end{tikzpicture}
    \caption{The cases dividing the possible player budgets $X_1, X_2 > 0$, which we borrow from~\cite{Kovenock_2012}. This illustration corresponds to a coalitional General Lotto game in which $\phi_1 = \phi_2$.}
    \label{fig:cases}
\end{figure}
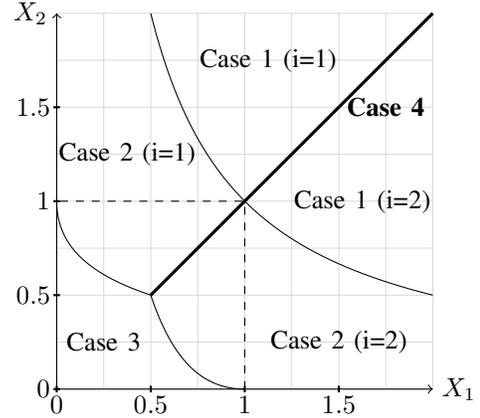

In line with our discussion in the previous sections, we wish to find the ranges of values $X_1$, $X_2$, $X_A$, $\phi_1$ and $\phi_2$ in which player 1 can improve its payoff by making a pre-commitment.
We first consider player $1$'s choice of using less than its full budget, which we show cannot increase its payoff in the theorem below.
\begin{thm}
    Consider the setting with player budgets $X_1, X_2, X_A > 0$ and battlefield values such that $\phi_i = \sum_{b \in \B_i} v_b > 0$ for $i \in \{1,2\}$.
    Then, $\pi_1(X_1-t, X^t_{A1}) \leq \pi_1(X_1, X_{A1})$, for all $X_1, X_2, t > 0$, where $X^t_{A1}$ is the adversary's allocation to $\B_1$ in its best response to player 1's transfer $t$, and $\pi_1(\cdot,\cdot)$ is given by the first entry of \eqref{eq:twoplayer_payoff}.
\end{thm}
\begin{proof}
    We omit the details of the proof for brevity. 
    However, the approach is to show that the derivative with respect to $t > 0$ of the function $\pi_1(X_1-t, X^t_{A1})$ is non-positive for all $X_1, X_2 > 0$.
    This can be verified independently for Cases 1 (i=1), 1 (i=2), 2 (i=1), 2 (i=2), 3, and 4.
\end{proof}

We have shown that player 1 cannot increase its payoff by pre-committing to using less than its full budget. 
Next, we analyze whether player 1 can increase its payoff by pre-committing a portion of its budget to a battlefield $b \in \B_1$, i.e., for some $t \in [0,X_1]$,
\[ \pi^t_1(X_1,X^t_{A1}) \geq \pi_1(X_1,X_{A1}), \]
where the functions $\pi^t_1$ and $\pi_1$ are defined over the set of battlefields $\B_1$ as in \eqref{eq:improvement}.
We denote by $X^t_{A1}$ the player $A$'s optimal allocation to the set of battlefields $\B_1$ when player 1 pre-commits. 
Formally, when player 1 pre-commits a force $t$ to $b \in \B_1$,
\[
X^t_{A1} = \begin{cases}
    X^\text{call}_{A1}(t) + t & \quad \text{if player $A$ calls}, \\
    X^\text{fold}_{A1}(t) & \quad \text{if player $A$ folds},
\end{cases}
\]
where $X^\text{call}_{A1}(t)$ and $X^\text{fold}_{A1}(t)$ are player $A$'s optimal allocations to $\B_1 \setminus \{b\}$ when calling and folding, respectively.
The value $X_{A1}$ is $A$'s optimal allocation to $\B_1$ when player 1 does not pre-commit.
In the following theorem, we characterize a region in the parameter space $0 < X_1, X_2 < 1$, $\phi_1, \phi_2 > 0$, and $0 < v_b < \phi_1$ within which a particular choice of pre-commitment increases player 1's payoff.

\begin{thm} \label{thm:three-player_precommitment}
    Consider a setting with player budgets $0 < X_1, X_2 < 1$ and $X_A = 1$, and two sets of battlefields $\B_1, \B_2$ such that $\phi_j = \sum_{b \in \B_j} v_b$ for $j \in \{1,2\}$.
    For a battlefield $b \in \B_1$ with value $v_b > 0$, the pre-commitment 
    \[ t^* = X_1 - \frac{\phi_1 - v_b}{\phi_2} X_2, \] 
    increases player 1's payoff if the budgets $X_1, X_2 > 0$ satisfy $t^* > 0$, and one of the following two sets of conditions:
    
    \vspace*{6pt}\noindent 1) player budgets $X_1, X_2$ are in Case 1 (i=1) for $\phi_1, \phi_2$, $X_1-t^*, X_2$ and $(X_1-t^*)/(1-t^*), X_2/(1-t^*)$ are in Case 2 (i=2) for $\phi_1-v_b, \phi_2$, and
    \begin{align}
        \phi_2(X_1{+}X_2{-}1) & \geq (\phi_1 {-} v_b) X_2(X_1{-}1) + v_b X_1X_2 \label{eq:3player1i} \\
        2 v_b &\geq \left[ X_1 {-} \frac{\phi_1 {-} v_b}{\phi_2}X_2 \right] \frac{\phi_2}{X_2}. \label{eq:3player1ii}
    \end{align}
    
    \vspace*{6pt}\noindent 2) player budgets $X_1, X_2$ are in Case 2 (i=1) for $\phi_1, \phi_2$, $X_1-t^*, X_2$ and $(X_1-t^*)/(1-t^*), X_2/(1-t^*)$ are in Case 2 (i=2) for $\phi_1-v_b, \phi_2$, and
    \begin{align}
        \phi_2(X_1{+}X_2{-}1) & \geq \sqrt{\phi_1\phi_2X_1X_2} - (\phi_1{-}v_b)X_2, \label{eq:3player2i}\\
        2 v_b &\geq \left[ X_1 {-} \frac{\phi_1{-}v_b}{\phi_2}X_2 \right] \frac{\phi_2}{X_2}. \label{eq:3player2ii}
    \end{align}
\end{thm}
The proof is provided in Appendix B.

In \cref{fig:phi3_05}, we plot the region in which the player budgets $0 < X_1, X_2 < 1$ satisfy the conditions in \cref{thm:three-player_precommitment} for two parameterizations of the game, i.e., the region in which player 1 can increase its payoff by pre-committing the force $t^* > 0$ to battlefield $b \in \B_1$. 
In contrast to the two-player setting, we observe that there is a nonempty set of parameters for which player 1 can make pre-commitments to increase its payoff in this three-player game setup.

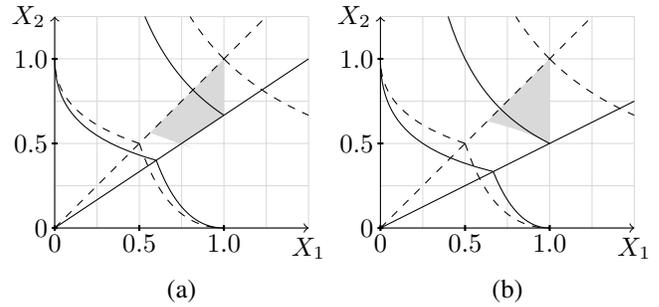
\begin{figure}[t]
    \centering
    \hspace{-25pt}
    \begin{minipage}[b]{0.2\textwidth}
    \begin{tikzpicture}[scale=2.25]
        \draw[very thin, gray!30, step=0.25 cm](0,0) grid (1.5,1.25);
        \draw[<->] (0,1.25) node[above, left] {$X_2$} -- (0,0) -- (1.5,0) node[below] {$X_1$};
        \foreach \x in {0,0.5,1.0}
            \draw[xshift=\x cm, thick] (0pt,-0.5pt)--(0pt,0.5pt) node[below] {$\x$};
        \foreach \x in {0,0.5,1.0}
            \draw[yshift=\x cm, thick] (-0.5pt,0pt)--(0.5pt,0pt) node[left] {$\x$};
        \fill[gray!30,domain=0.8165:1,variable=\x] (1,1) -- (0.8165,0.8165)
            -- plot (\x,{2/3/\x}) -- (1,2/3) -- cycle;
        \fill[gray!30,domain=0.563:0.75,variable=\x] (0.563,0.563)
            -- plot ({1/(1+2*\x-1.22474*\x^(1/2))},{\x/(1+2*\x-1.22474*\x^(1/2))})
            -- (0.75,0.5) -- (1.0,2/3) -- (0.8165,0.8165) -- cycle;
        \draw[dashed,domain=0:1.25] plot (\x,\x);
        \draw[dashed,domain=0.5:1] plot (\x, {(1-\x)^2/\x});
        \draw[dashed,domain=1:1.5] plot (\x, {1/\x});
        \draw[dashed,domain=0.5:1] plot ({(1-\x)^2/\x}, \x);
        \draw[dashed,domain=1:1.25] plot ({1/\x}, \x);
        \draw[domain=0:1.5] plot (\x,{2*\x/3});
        \draw[domain=0.533:1] plot (\x,{2/3/\x});
        \draw[domain=0.6:1] plot (\x,{-1.5*(-\x*\x+2*\x-1)/\x});
        \draw[domain=0.4:1] plot ({2*(\x-1)*(\x-1)/\x/3},\x);
        \node[below] at (0.75,-0.25) {(a)};
    \end{tikzpicture}
    \end{minipage}
    \hspace{15pt}
    \begin{minipage}[b]{0.2\textwidth}
    \begin{tikzpicture}[scale=2.25]
        \draw[very thin, gray!30, step=0.25 cm](0,0) grid (1.5,1.25);
        \draw[<->] (0,1.25) node[above, left] {$X_2$} -- (0,0) -- (1.5,0) node[below] {$X_1$};
        \foreach \x in {0,0.5,1.0}
            \draw[xshift=\x cm, thick] (0pt,-0.5pt)--(0pt,0.5pt) node[below] {$\x$};
        \foreach \x in {0,0.5,1.0}
            \draw[yshift=\x cm, thick] (-0.5pt,0pt)--(0.5pt,0pt) node[left] {$\x$};
        \fill[gray!30,domain=0.707:1,variable=\x] (1,1) -- (0.707,0.707)
            -- plot (\x,{1/(2*\x)}) -- (1,0.5) -- cycle;
        \fill[gray!30,domain=0.6306:1,variable=\x] (0.6306,0.6306)
            -- plot (\x,{(-\x+(\x*(4-3*\x))^(1/2)+2)/4}) -- (1,0.5) -- (0.75,0.75) -- cycle;
        \draw[dashed,domain=0:1.25] plot (\x,\x);
        \draw[dashed,domain=0.5:1] plot (\x, {(1-\x)^2/\x});
        \draw[dashed,domain=1:1.5] plot (\x, {1/\x});
        \draw[dashed,domain=0.5:1] plot ({(1-\x)^2/\x}, \x);
        \draw[dashed,domain=1:1.25] plot ({1/\x}, \x);
        \draw[domain=0:1.5] plot (\x,{\x/2});
        \draw[domain=0.4:1] plot (\x,{1/(2*\x)});
        \draw[domain=0.667:1] plot (\x,{2*(\x*\x-2*\x+1)/\x});
        \draw[domain=0.333:1] plot ({(\x*\x-2*\x+1)/(2*\x)},\x);
        \node[below] at (0.75,-0.25) {(b)};
    \end{tikzpicture}
    \end{minipage}
    \caption{The region in which  player 1 can increase its payoff by pre-committing $t^* > 0$ of its force to the battlefield $b \in \B_1$, as in \cref{thm:three-player_precommitment}. 
    Plot a) corresponds to parameters $\phi_1 - v_b = \phi_2 = 1$, and $v_b = 1/2$.
    Plot b) corresponds to parameters $\phi_1 - v_b = \phi_2 = 1$, and $v_b = 1$.
    The solid lines divide the $X_1,X_2$-plane into the various cases for the game without pre-commitments, i.e., the General Lotto games are over battlefields $\B_1$ and $\B_2$. The dashed lines correspond to the game in which player 1 pre-commits a force $t^* > 0$ to battlefield $b \in \B$, i.e, the General Lotto games are over battlefields in $\B_1  \setminus \{b\}$ and $\B_2$.}
    \label{fig:phi3_05}
\end{figure}

\section{Conclusion and future work}
In this paper, we investigated whether announcing a strategic intention could ever benefit a player. We modelled a variant of the Colonel Blotto game in which one player can pre-commit a portion of its budget. 

In the one-vs-one setting, we proved that there is never an incentive for a weaker player to pre-commit a portion of its budget to any given battlefield.
In the three-player setting, wherein the adversary faces two players at once, we demonstrated that situations \textit{do} exist in which one of the players facing the adversary can improve its payoff by pre-committing some of its forces to a battlefield. In effect, such announcements shift the adversary's allocation of resources towards the other player.  

Future directions include an investigation of emergent equilibrium outcomes when multiple players facing a common adversary have the option to announce a pre-commitment. 
Additionally, exploring how a pre-commitment made by player $1$ can affect the payoffs experienced by the adversary and player $2$ will provide insight into the consequences of self-interested play.

\bibliographystyle{IEEEtran}
\bibliography{sources}

\appendix
\subsection{The Two-Player Setting}

\begin{proof}[Proof of Theorem \ref{theorem:twoplayer}]
    Suppose player $B$ pre-commits to sending $t \in (0,X_B]$ of its budget to a battlefield $b\in\B$.
    Define $\gamma:=\frac{X_B}{X_A} < 1$ as the budget ratio. The approach is to show that for any set of parameters $(\gamma,v_b) \in (0,1) \times (0,\phi)$, the payoff difference is negative (i.e., $\Delta\pi_B^t < 0$). 
    This amounts to identifying the regions where it is optimal for player A to call or fold, and evaluating the quantity $\Delta\pi_B^t$. 
    
    Let $\tilde{\phi} := \phi - v_b$ be the total valuation of prizes without $b$. 
    Using the payoffs for calling and folding (\eqref{eq:call_payoff} and \eqref{eq:fold_payoff}), the adversary calls on the pre-commitment if  
    \begin{equation}
        C(t) := \frac{\tilde{\phi}}{X_A^2} t^2 - \frac{t}{X_A}(2 v_b + \tilde{\phi}\gamma) + 2v_b  > 0.
    \end{equation}
    It will fold if $C(t) < 0$, and is indifferent between calling and folding if $C(t) = 0$. 
    We have $C(0) = 2v_b > 0$, $C(X_B) = 2v_b(1 - \gamma) > 0$, and $C'(t_m) = 0$ where $t_m := \frac{v_b}{\tilde{\phi}}X_A + \frac{1}{2}X_B$. 
    The roots of $C(t)$ are
    \begin{equation}
            t_{\pm} := t_m \pm \frac{X_A}{\tilde{\phi}}\sqrt{Q(v_b) }.
    \end{equation}
    where $Q(v_b) := \frac{\tilde{\phi}^2}{X_A^2} t_m^2 - 2v_b\tilde{\phi}$. Consequently, we will have the following cases.
    \begin{enumerate}
        \item  Suppose either $t_m < X_B$ and $Q(v_b) < 0$, or $t_m > X_B$. Then $C(t) > 0$ for all $t\in[0,X_B]$, and the adversary calls on all pre-commitments.
        \item Suppose $t_m < X_B$ and $Q(v_b) > 0$. Then the adversary calls if $t \in [0,t_-) \cup (t_+,X_B]$ and folds if $t \in [t_-,t_+]$
    \end{enumerate}
    The regions described by these cases is depicted in Figure \ref{fig:thm1_proof}. 
    To prove the result, we need to show  $\Delta\pi_B^t< 0$ for any set of parameters. Suppose the adversary calls on the precommitment under any case. Player $B$'s payoff difference is
    \begin{equation}
        \begin{aligned}
            \Delta\pi_B^t &= \tilde{\phi}\left( \frac{X_B - t}{X_A-t} - 1 \right) - v_b - \phi\left( \frac{X_1}{X_A} - 1\right) \\
            &= \tilde{\phi}\frac{X_B - t}{X_A-t} - \phi \frac{X_B}{X_A} \\
            &< 0, \quad \forall t \in [0,X_B].
        \end{aligned}
    \end{equation}
    Consequently, player $B$ never has an incentive to pre-commit when the adversary calls. It remains to be shown that $B$ also never has an incentive to pre-commit in cases where the adversary folds. We thus focus on case 2) when $t \in [t_-,t_+]$. In this regime, the following condition is satisfied:
    \begin{equation}
        \begin{aligned}
            Q(v_b) &= v_b^2(3-\gamma + \frac{\gamma^2}{4}) + v_b\phi(-2+\gamma-\frac{\gamma^2}{2}) + \frac{\gamma^2}{4}\phi^2 \\
            &> 0.
        \end{aligned}
    \end{equation}
    We have $Q(0) = \frac{1}{4}\gamma^2\phi > 0$ and $Q(\phi) = \phi^2(1 - \frac{1}{4}\gamma^2) + \frac{\gamma^2}{4}\phi > 0$. The roots of $Q(v_b)$ are calculated to be
    \begin{equation}
        v_\pm(\gamma) := \frac{\phi}{2} \frac{(2-\gamma+\frac{1}{2}\gamma^2) \pm 2\sqrt{1-\gamma}}{3-\gamma+\frac{1}{4}\gamma^2}
    \end{equation}
    The minimizer of $Q(v_b)$ lies in $(0,\phi)$. Additionally, the roots $v_\pm$ are real and also lie in the interval $(0,\phi)$. We then have $Q(v_b) > 0$ for $v_b \in [0,v_-) \cup (v_+,\phi]$. 
    
    We show case 2) further restricts attention to $v_b \in [0,v_-)$. Note the condition $t_m < X_B$ is equivalent to  $v_b < H(\gamma) := \phi \frac{\gamma/2}{1 + \gamma/2}$. This function monotonically increases from zero at $\gamma=0$ to $\frac{\phi}{3}$ at $\gamma=1$. We therefore have $v_+(\gamma) > H(\gamma)$, since  $v_+(\gamma)$ monotonically decreases from $\frac{2\phi}{3}$ at $\gamma=0$ to $\frac{\phi}{3}$ at $\gamma = 1$. Now, we show $v_-(\gamma) < H(\gamma)$. It is true that $H(\gamma) > \frac{\phi}{3}\gamma$, as well as
    \begin{equation}
        \begin{aligned}
            v_-(\gamma) &= \phi\frac{1-\frac{\gamma}{2} + \frac{\gamma^2}{4} - \sqrt{1-\gamma}}{3-\gamma+\frac{\gamma^2}{4}} \\
            &< \phi\frac{\frac{\gamma}{2} + \frac{\gamma^2}{4}}{3-\gamma+\frac{\gamma^2}{4}} = \phi\gamma \frac{\frac{1}{2} + \frac{\gamma}{4}}{3-\gamma(1-\frac{\gamma}{4})} \\
            &< \frac{\phi}{3}\gamma
        \end{aligned}
    \end{equation}
    The first inequality is due to $\sqrt{1-\gamma} > 1-\gamma$, and the second inequality follows because the numerator (denominator) of the second line is increasing (decreasing) in $\gamma \in (0,1)$. Therefore, $v_-(\gamma) < H(\gamma)$.

  When the adversary folds, player $B$ wins battlefield $v_b$ and its payoff difference is
    \begin{equation}\label{eq:diff_fold}
        \begin{aligned}
            \Delta\pi_B^t &= \tilde{\phi}\left( \frac{X_B - t}{X_A} - 1 \right) + v_b - \phi\left( \frac{X_B}{X_A} - 1\right) \\
            &= v_b(2 - \frac{X_B}{X_A}) - \frac{\phi - v_b}{X_A}t, \quad t \in [t_-,t_+]. 
        \end{aligned}
    \end{equation}
    Since \eqref{eq:diff_fold} is decreasing in $t$, it suffices to show it is negative for $t = t_-$ and for all $v_b \in (0,v_-)$: 
    \begin{equation}
        \begin{aligned}
            \Delta\pi_B^{t_-} &= v_b(2 - \gamma) - \left(v_b + \frac{\phi-v_b}{2}\gamma - \sqrt{Q(v_b)} \right)
        \end{aligned}
    \end{equation}
  This function is decreasing with respect to $v_b$:
  \begin{equation}\label{eq:decreasing}
      \frac{\partial\Delta\pi_B^{t_-}}{\partial v_b} = 1-\frac{\gamma}{2} + \frac{Q'(v_b)}{2\sqrt{Q(v_b)}} < 0.
  \end{equation}
  To verify, we know that $Q(v_b) > 0$ and $Q'(v_b) < 0$ for all $v_b \in (0,v_-)$. Observe that $(1-\frac{\gamma}{2}) + \frac{Q'(0)}{2\sqrt{Q(0)}}  = 2 - \frac{2}{\gamma} - \gamma < 0$.    Additionally, $\frac{\partial}{\partial v_b}\left[ \frac{Q'(v_b)}{\sqrt{Q(v_b)}} \right] < 0$ if and only if $\frac{1}{2}(Q'(v_b))^2 - Q''(v_b)Q(v_b) > 0$. A calculation shows this condition is equivalent to $\gamma < 1$. Thus, \eqref{eq:decreasing} holds. Because $\Delta\pi_B^{t_-} = 0$ at $v_b = 0$, we have $\Delta\pi_B^{t_-} < 0$ for all $v_b \in (0,v_-)$. This concludes the proof.
\end{proof}
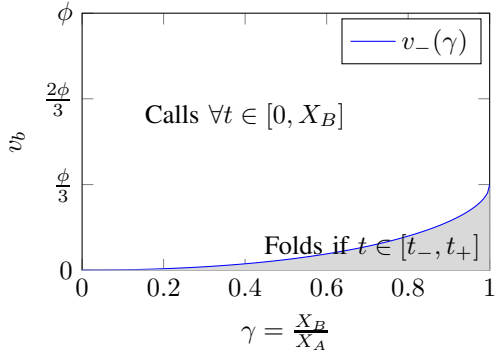
\begin{figure}
    \centering
    \begin{tikzpicture}
        \begin{axis}[%
        xlabel={$\gamma = \frac{X_B}{X_A}$},
        ylabel={$v_b$},
        xmin=0, xmax=1,
        ymin=0, ymax=1,
        ytick={0,0.333,0.667,1},
        yticklabels={$0$,$\frac{\phi}{3}$,$\frac{2\phi}{3}$,$\phi$},
        height=5cm,
        ]
        \addplot+ [name path=A, no marks, domain=0:1, samples=201] {1/2*((2-x+x*x/2)-2*(1-x)^(1/2))/(3-x+x*x/4)};
        \addlegendentry{$v_{-}(\gamma)$}
        \addplot+ [name path=B, no marks, gray!30, domain=0:1] {0};
        \addplot [gray!30] fill between [of=A and B];
        \node at (0.4,0.6) {Calls $\forall t \in [0, X_B]$};
        \node[above left] at (1,0) {Folds if $t \in [t_{-}, t_{+}]$};
        \end{axis}
    \end{tikzpicture}
    \caption{Player A's optimal decisions over the entire parameter space $(\gamma,v_b) \in (0,1) \times (0,\phi)$ in the two player setting. The region above the curve $v_-(\gamma)$ corresponds to case 1 in the proof of Theorem \ref{theorem:twoplayer}. The region under the curve  corresponds to case 2. }
    \label{fig:thm1_proof}
\end{figure}
As Figure \ref{fig:thm1_proof} depicts, player $A$ folds only if player $B$ pre-commits a sufficient (but not too high) amount of force $t$ to a battlefield $b$ that has a relatively low value. Even so, player $B$ does not obtain a higher payoff by forcing A to fold because the value of battlefield $b$ does not compensate for the high pre-commitment required.

\subsection{The Three-Player Setting}

\begin{proof}[Proof of Theorem \ref{thm:three-player_precommitment}]
    The proof amounts to showing that in the regions specified, player 1 receives higher payoff if it pre-commits to the transfer $t^* > 0$ than if it were to play the original game, under the assumption that player $A$ calls. Due to this assumption, we must also show that player $A$ prefers to call on the pre-commitment rather than fold.
    Recall that when player $A$ calls, we re-normalize the player budgets such that $X_1 - t \to \frac{X_1-t}{1-t}$, $X_2 \to \frac{X_2}{1-t}$, and $X_A - t \to 1$.
    Note that, if player 1 pre-commits a force just below $t^* > 0$, the pairs $(X_1-t^*, X_2)$ and $(\frac{X_1-t^*}{1-t^*}, \frac{X_2}{1-t^*})$ are in Case 2 ($i=2$) of the resulting game on battlefields $\B_1 \setminus \{b\}$ and $\B_2$.
    We split the remainder of the proof into the two different regions specified in the claim.
    
    \vspace*{6pt}\noindent{\it Part 1):} When the values $X_1, X_2$ are in Case 1 (i=1), and player 1 does not pre-commit, the payoff to player 1 is
    \[ \pi_1(X_1, X_{A1}) = \phi_1 \left[ X_1 - 1 \right]. \]
    When player 1 pre-commits a force of $t > 0$ to battlefield $b \in \B_1$ and player $A$ calls, its payoff is
    \[ \pi^t_1(X_1, X^\text{call}_{A1}(t)) = (\phi_1-v_b) \left[ 1 - \frac{1-\sqrt{\frac{\phi_2\frac{X_1-t}{1-t}\frac{X_2}{1-t}}{\phi_1-v_b}}}{\frac{X_1-t}{1-t}} \right] - v_b. \]
    Player 1 prefers to pre-commit a force of $t > 0$ if $\pi^t_1(X_1,X^\text{call}_{A1}(t)) \geq \pi_1(X_1,X_{A1})$.
    For $t = t^*$, this condition simplifies to \eqref{eq:3player1i}, i.e.,
    \[ \phi_2(X_1+X_2-1) \geq (\phi_1-v_b) X_2(X_1-1) + v_b X_1X_2. \]
    
    When player 1 pre-commits $t > 0$, the payoff to player $A$ for folding is
    \begin{align*}
        u^{\text{fold}}_A(X^\text{call}_{A1}(t),X_1,t) =& (\phi_1{-}v_b) \!\left[ \frac{X^\text{fold}_{A1}(t)}{X_1{-}t} {-} 1 \right] \\
        &\hspace{25pt}{+} \phi_2 \! \left[ \frac{1{-}X^\text{fold}_{A1}(t)}{X_2} {-} 1 \right] \!{-} v_b. 
    \end{align*}
    If player $A$ calls, its payoff is
    \begin{align*}
        u^{\text{call}}_A(X^\text{call}_{A1}(t),X_1,t) = &(\phi_1-v_b) \!\left[ \frac{X^\text{call}_{A1}(t)}{\frac{X_1-t}{1-t}} - 1 \right] \\
        &\hspace{25pt} + \phi_2 \left[ \frac{1{-}X^\text{call}_{A1}(t)}{\frac{X_2}{1{-}t}} {-} 1 \right] \!{+} v_b.
    \end{align*}
    For given player budgets $X_1, X_2 > 0$ and pre-commitment $t \in [0,X_1]$, player $A$ prefers to call when 
    \[ u^{\text{call}}_A(X^\text{call}_{A1}(t),X_1,t) \geq u^{\text{fold}}_A(X^\text{fold}_{A1}(t),X_1,t). \]
    For $t = t^*$, this condition simplifies to \eqref{eq:3player1ii},
    \[ 2 v_b \geq \left[ X_1 - \frac{\phi_1-v_b}{\phi_2}X_2 \right] \frac{\phi_2}{X_2}. \]
    
    \vspace*{6pt}\noindent{\it Part 2):} When the values $X_1, X_2$ are in Case 2 (i=1), and player 1 does not pre-commit a force, the payoff to player 1 is
    \[ \pi_1(X_1, X_{A1}) = \phi_1 \left[ \frac{X_1}{\sqrt{\frac{\phi_1X_1X_2}{\phi_2}}} - 1 \right]. \]
    When player 1 pre-commits a force of $t > 0$ to battlefield $b \in \B_1$ and player $A$ calls, its payoff is
    \[ \pi^t_1(X_1, X^\text{call}_{A1}(t)) = (\phi_1-v_b) \left[ 1 - \frac{1-\sqrt{\frac{\phi_2\frac{X_1-t}{1-t}\frac{X_2}{1-t}}{\phi_1-v_b}}}{\frac{X_1-t}{1-t}} \right] - v_b. \]
    Player 1 prefers to pre-commit a force of $t > 0$ if $\pi^t_1(X_1,X^\text{call}_{A1}(t)) \geq \pi_1(X_1,X_{A1})$.
    For $t = t^*$, this condition simplifies to \eqref{eq:3player2i},
    \[ \phi_2(X_1+X_2-1) \geq \sqrt{\phi_1\phi_2X_1X_2} - (\phi_1-v_b)X_2 . \]
    
    Player $A$'s payoffs for calling and folding are the same as in the previous part, so \eqref{eq:3player2ii} must be satisfied.
\end{proof}

\end{document}